\documentclass[letterpaper, 10 pt, conference]{ieeeconf}  
\usepackage{graphicx}
\usepackage{amsmath}
\usepackage{amssymb}

\usepackage{amsthm}
\usepackage{url}
\usepackage{color}
\usepackage{multirow}
\usepackage{hhline}
\usepackage{amsfonts}
\usepackage{hyperref}
\usepackage{subfigure}  
\usepackage{bbm}

\newtheorem{remark}{Remark}
\newtheorem{definition}{Definition}
\newtheorem{proposition}{Proposition}
\newtheorem{theorem}{Theorem}

\newtheorem{problem}{Problem}
\newtheorem{assumption}{Assumption}

\IEEEoverridecommandlockouts                              
\title{\LARGE \bf
Consensus Control for Leader-follower Multi-agent Systems under Prescribed Performance Guarantees 
}

\author{Fei Chen and Dimos V. Dimarogonas
\thanks{This work was supported by the EU H2020 Co4Robots
Project, the Swedish Research Council (VR) and the Knut och Alice Wallenberg Foundation (KAW).}
\thanks{Fei Chen and Dimos V. Dimarogonas are with the Division of Decision and Control Systems, KTH Royal Institute of Technology, SE-100 44 Stockholm, Sweden
        {\tt\small \{fchen,dimos\}@kth.se}}%
}

\begin{document}

\maketitle
\thispagestyle{empty}
\pagestyle{empty}

\begin{abstract}
This paper addresses the problem of distributed control for leader-follower multi-agent systems under prescribed performance guarantees. Leader-follower is meant in the sense that a group of agents with external inputs are selected as leaders in order to drive the group of followers in a way that the entire system can achieve consensus within certain prescribed performance transient bounds. Under the assumption of tree graphs, a distributed control law is proposed when the decay rate of the performance functions is within a sufficient bound. Then, two classes of tree graphs that can have additional followers are investigated. Finally, several simulation examples are given to illustrate the results.

\end{abstract}

\section{INTRODUCTION}
The consensus problem has attracted great interest due to its wide applications in  robotics, cooperative control~\cite{fax2003information}, formation~\cite{balch1998behavior} and flocking~\cite{tanner2007flocking}. Consensus or agreement is achieved when a group of agents converge to a common value. The first order consensus protocol was first introduced in~\cite{olfati2004consensus}, where the authors discussed the consensus problem of directed and undirected graphs with fixed or switching topologies and time delays. Second order consensus protocol has been investigated in~\cite{ren2007distributed}, where the states
of the agents converge to a constant or a linear function.     

In this work, we study the consensus problem in a leader-follower framework, that is, one or more agents are selected as \emph{leaders} with external inputs in addition to the first order consensus protocol. The remaining agents are \emph{followers} only obeying the first order consensus protocol. Recent research that has been done in the leader-follower framework can be divided into two parts. The first part deals with the controllability of leader-follower multi-agent systems. For instance, controllability of networked systems was first investigated in~\cite{tanner2004controllability} by deriving conditions on the network topology, which ensures that the network can be controlled by a particular member which acts as a leader. In~\cite{egerstedt2012interacting,rahmani2009controllability}, the authors identify necessary conditions for the controllability of the corresponding leader-follower networks using equitable partitions of graphs. Controllability conditions for leader-follower multi-agent systems with double integrator dynamics and their connection with graph topology properties are addressed in~\cite{goldin2010controllability}. The second part targets leader selection problems~\cite{yaziciouglu2013leader,patterson2010leader,franchi2018online}. These involve the problem of how to choose the leaders among the agents such that the leader-follower system satisfies the requirements such as controllability, optimal performance or formation maintenance.

Prescribed performance control (PPC) was originally proposed in~\cite{bechlioulis2008robust}, with the aim to prescribe the evolution of system output or the tracking error within some predefined region. For example, an agreement protocol that can additionally achieve prescribed performance for a combined error of positions and velocities is designed in~\cite{macellari2017multi} for multi-agent systems with double integrator dynamics, while PPC for multi-agent average consensus with single integrator dynamics is presented in~\cite{karayiannidis2012multi}. In~\cite{bechlioulis2014robust}, the authors consider the formation control problem for nonlinear multi-agent systems with prescribed performance guarantees and connectivity constraints. Funnel control, which uses a similar idea as PPC was first introduced in~\cite{ilchmann2005tracking} for reference tracking. In~\cite{berger2018funnel}, the authors utilize funnel control for uncertain nonlinear systems that have arbitrary strict relative degree and input-to-state stable internal dynamics.


In this work, we are interested in how to design control strategies for the leaders such that the leader-follower multi-agent system achieves consensus within certain performance bounds. Compared with existing work of PPC for multi-agent systems~\cite{macellari2017multi}, we apply a PPC law only to the leaders while most of the related work, including~\cite{macellari2017multi}, applies PPC to all the agents to achieve consensus. The benefit of this work is to lower the cost and control effort since the followers will follow the leaders by obeying first order consensus protocols without any additional control. Unlike other leader-follower consensus approaches using PPC~\cite{katsoukis2016output}, in which the multi-agent system only has one leader and the leader is treated as a reference for the followers, we focus on a more general framework in the sense that we can have more than one leader and the leaders are designed in order to steer the entire system achieving consensus within the prescribed performance bounds. The difficulties in this work are due to the combination of uncertain topologies, leader amount and leader positions. In addition, the leader can only communicate with its neighbouring agents. The contributions of the paper can be summarized as: i) within this general leader-follower framework, under the assumption of tree graphs, a distributed control law is proposed when the decay rate of the performance functions is within a sufficient bound; ii) the specific classes of chain and star graphs that  can  have  additional  followers are investigated. 

The rest of the paper is organized as follows. In Section II, preliminary knowledge is introduced and the problem is formulated, while Section
III presents the main results, which are further verified by
simulation examples in Section IV. Section V closes with concluding remarks and future work.


\section{PRELIMINARIES AND PROBLEM STATEMENT}

\subsection{Graph Theory} 
An undirected graph~\cite{mesbahi2010graph} $\mathcal{G}=(\mathcal{V},\mathcal{E})$ comprises of the vertices set $\mathcal{V}=\{1,2,\dots,n\}$ and the edges set $\mathcal{E}=\{(i,j)\in \mathcal{V}\times \mathcal{V}\mid j\in \mathcal{N}_i\}$ indexed by $e_1,e_2,\dots, e_m$. Here, $m=|\mathcal{E}|$ is the number of edges and $\mathcal{N}_i$ denotes the agents in the neighbourhood of agent $i$ that can communicate with $i$. The \emph{adjacency matrix} $\mathbb{A}$ of $\mathcal{G}$ is the $n\times n$ symmetric matrix whose elements $a_{ij}$ are given by $a_{ij}=1$, if $(i,j)\in \mathcal{E}$, and $a_{ij}=0$, otherwise. The degree of vertex $i$ is defined as $d_i={\sum}_{j\in \mathcal{N}_i}a_{ij}$. Then the \emph{degree matrix} is $\Delta=\mbox{diag}(d_1, d_2,\dots,d_n)$. The \emph{graph Laplacian} of $\mathcal{G}$ is $L=\Delta-\mathbb{A}$. A \emph{path} is a sequence of edges connecting two distinct vertices. A graph is \emph{connected} if there exists a path between any pair of vertices. By assigning an orientation to each edge of $\mathcal{G}$ we can define the \emph{incidence matrix} $D=D(\mathcal{G})=[d_{ij}]\in \mathbb{R}^{n\times m}$. The rows of $D$ are indexed by the vertices and the columns are indexed by the edges with $d_{ij}=1$ if the vertex $i$ is the head of the edge $(i,j)$, $d_{ij}=-1$ if the vertex $i$ is the tail of the edge $(i,j)$ and $d_{ij}=0$ otherwise. Based on the incidence matrix, the graph Laplacian of $\mathcal{G}$ can be described as $L=DD^T$. In addition, $L_e=D^TD$ is the so called \emph{edge Laplacian}~\cite{mesbahi2010graph} and $c_{ij}$ denotes the elemnts of $L_e$. 

\subsection{System Description}
In this work, we consider a multi-agent system with vertices  $\mathcal{V}=\{1,2,\dots,n\}$. Without loss of generality, we suppose that the first $n_f$ agents are selected as followers while the last $n_l$ agents are selected as leaders with respective vertices set $\mathcal{V}_F=\{1,2,\dots,n_f\}$,  $\mathcal{V}_L=\{n_f+1,n_f+2,\dots,n_f+n_l\}$ and $n=n_f+n_l$.  

Let $x_i\in \mathbb{R}$ be the position of agent $i$, where we only consider the one dimensional case, without loss of generality. Specifically, the results can be extended to higher dimensions with appropriate use of the Kronecker product. The state evolution of each follower $i\in \mathcal{V}_F$ is governed by the first order agreement protocol:
\begin{equation}\label{eq:fdynamic}
\dot{x}_i=\sum\limits_{j\in \mathcal{N}_i}(x_j-x_i),    
\end{equation}
while the state evolution of each leader $i\in \mathcal{V}_L$ is governed by the first order agreement protocol with an assigned external input $u_i\in \mathbb{R}$:
\begin{equation}\label{eq:ldynamic}
\dot{x}_i=\sum\limits_{j\in \mathcal{N}_i}(x_j-x_i)+u_i.   
\end{equation}

Let $x=[x_1,\dots, x_{n_f},\dots,x_n]^T\in \mathbb{R}^n$ be the stack vector of absolute positions of all the agents and $u=[u_1,\dots,u_{n_l}]^T\in \mathbb{R}^{n_l}$ be the control input vector . Denote $\bar{x}=[\bar{x}_1,\dots,\bar{x}_m]^T$ as the stack vector of relative positions between the pair of communicating agents $(i,j)\in \mathcal{E}$, where $\bar{x}_k\triangleq x_{ij}=x_i-x_j, k=1,2,\dots, m$. It can be easily verified that $Lx=D\bar{x}$ and $\bar{x}=D^Tx$. In addition, if $\bar{x}=0$, we have that $Lx=0$. By stacking \eqref{eq:fdynamic} and \eqref{eq:ldynamic}, the dynamics of the leader-follower multi-agent system is rewritten as:
\begin{equation}\label{eq:nodedynamic}
  \Sigma:  \dot{x}=-Lx+Bu,
\end{equation}
where $L$ is the graph Laplacian and $B=\left[
\begin{smallmatrix} 0_{n_f\times n_l}\\ I_{n_l}
\end{smallmatrix}\right] .$
\subsection{Prescribed Performance Control}
The aim of PPC is to prescribe the evolution of the system output or the tracking error within some predefined region described as follows:
\begin{equation}\label{eq:pfp}
-M_{ij}\rho_{ij}(t)<x_{ij}(t)<\rho_{ij}(t)\mbox{    \hspace{5mm}  if  }x_{ij}(0)>0
\end{equation}
\begin{equation}\label{eq:pfn}
-\rho_{ij}(t)<x_{ij}(t)<M_{ij}\rho_{ij}(t)\mbox{    \hspace{5mm}  if  }x_{ij}(0)<0
\end{equation}
$\rho_{ij}(t):\mathbb{R}_+\rightarrow \mathbb{R}_+\setminus \{0\}$ are positive, smooth and stirctly decreasing performance functions that introduce the predefined bounds for the target system outputs or the tracking errors. One example choice is $\rho_{ij}(t)=(\rho_{ij0}-\rho_{ij\infty})e^{-l_{ij}t}+\rho_{ij\infty}$ with $\rho_{ij0},\rho_{ij\infty}$ and $l_{ij}$ positive parameters and $\rho_{ij\infty}=\text{lim}_{t\rightarrow\infty}\rho_{ij}(t)$ represents the maximum allowable tracking error at the steady state; $M_{ij}$ represents the maximum allowed overshot.     

Normalizing $x_{ij}(t)$ with respect to the performance function $\rho_{ij}(t)$, we define the modulated error as $\hat{x}_{ij}(t)$ and the corresponding prescribed performance region $\mathcal{D}_{ij}$:
\begin{equation}\label{eq:modout}
\hat{x}_{ij}(t)=\frac{x_{ij}(t)}{\rho_{ij}(t)}
\end{equation}

\begin{equation}\label{eq:modpfp}
\mathcal{D}_{ij}\triangleq \{\hat{x}_{ij}:\hat{x}_{ij}\in (-M_{ij},1)\}\mbox{\hspace{5mm}  if  }x_{ij}(0)>0
\end{equation}
\begin{equation}\label{eq:modpfn}
\mathcal{D}_{ij}\triangleq \{\hat{x}_{ij}:\hat{x}_{ij}\in (-1,M_{ij})\}\mbox{\hspace{5mm}  if  }x_{ij}(0)<0
\end{equation}
Then the modulated error is transformed through the transformed function $T_{ij}$ that defines the smooth and strictly increasing mapping $T_{ij}: \mathcal{D}_{ij}\rightarrow \mathbb{R}$ and $T_{ij}(0)=0$. One example choice is \begin{equation}
    T_{ij}(\hat{x}_{ij})=\ln \left(-\frac{\hat{x}_{ij}+1}{\hat{x}_{ij}-M_{ij}}\right). 
\end{equation}
Hence the transformed error is defined as 
\begin{equation}\label{eq:tranerr}
\varepsilon_{ij}(\hat{x}_{ij})=T_{ij}(\hat{x}_{ij}) 
\end{equation}
It can be verified that if the transformed error $\varepsilon_{ij}(\hat{x}_{ij})$ is bounded, then the modulated error $\hat{x}_{ij}$ is constrained within the regions \eqref{eq:modpfp}, \eqref{eq:modpfn}. This also implies the error $x_{ij}$ evolves within the predefined performance bounds \eqref{eq:pfp} and \eqref{eq:pfn}, respectively.
Differentiating \eqref{eq:tranerr} with respect to time, we derive
\begin{equation}\label{eq:diff}
\dot{\varepsilon}_{ij}(\hat{x}_{ij})=\mathcal{J}_{T_{ij}}(\hat{x}_{ij},t)[\dot{x}_{ij}+\alpha_{ij}(t)x_{ij}]
\end{equation}
where
\begin{equation}\label{eq:jacobian}
\mathcal{J}_{T_{ij}}(\hat{x}_{ij},t)\triangleq \frac{\partial T_{ij}(\hat{x}_{ij})}{\partial \hat{x}_{ij}}\frac{1}{\rho_{ij}(t)}>0
\end{equation}
\begin{equation}\label{eq:alpha}
\alpha_{ij}(t)\triangleq -\frac{\dot{\rho}_{ij}(t)}{\rho_{ij}(t)}>0
\end{equation}
are the normalized Jacobian of the transformation function $T_{ij}$ and the normalized derivative of the performance function, respectively.  
\subsection{Problem Statement}
In this work, we are interested in how to design a control strategy for the leader-follower multi-agent system given by \eqref{eq:nodedynamic} such that the controlled system can achieve consensus within the prescribed performance requirements. The control strategy is only applied to the leaders and these drive the followers to guarantee the entire multi-agent system meet the requirements. Formally,  

\begin{problem}
Let the leader-follower multi-agent system $\Sigma$ defined by \eqref{eq:nodedynamic} with the communication graph $\mathcal{G}=(\mathcal{V},\mathcal{E})$ and the prescribed performance functions $\rho_{ij}, (i,j)\in \mathcal{E}$. Derive a control strategy such that the controlled leader-follower multi-agent system achieves consensus within $\rho_{ij}$.
\end{problem}

\section{MAIN RESULTS}
In this section, we design the control for the leader-follower multi-agent system \eqref{eq:nodedynamic} such that the system can achieve consensus within the prescribed performance functions
\begin{equation}\label{eq:performancefunction} 
\rho_{ij}(t)=(\rho_{ij0}-\rho_{ij\infty})e^{-l_{ij}t}+\rho_{ij\infty}.
\end{equation}
Here the performance functions are chosen as \eqref{eq:performancefunction} without loss of generality and the communication agents share information about their performance functions and transformation functions, that is, $\rho_{ij}(t)=\rho_{ji}(t), M_{ij}=M_{ji}$ and $T_{ij}(\hat{x}_{ij})=-T_{ji}(\hat{x}_{ji})$. This means the communication between the neighbouring agents are bidirectional and the graph $\mathcal{G}$ is assumed undirected. 

Consensus is achieved in the sense that the stack vector $\bar{x}$ of relative positions converges to zero as $t\rightarrow \infty$. We then rewrite the dynamics of the leader-follower multi-agent system \eqref{eq:nodedynamic} into the edge space in order to characterise the dynamics of the relative positions. We first rewrite \eqref{eq:nodedynamic} into the dynamics corresponding to followers and leaders, respectively. The corresponding incidence matrix is denoted as $D=\begin{bmatrix}
D_f^T&D_i^T
\end{bmatrix}^T$ with $D_f, D_i$ denoting the incidence matrices that characterise how followers and leaders are connected with other agents. Then \eqref{eq:nodedynamic} is reorganised as 
\begin{equation}\label{eq:DS_re}
\Sigma:\begin{bmatrix}
\dot{x}_f\\ \dot{x}_l
\end{bmatrix}=\begin{bmatrix}
A_f&B_f\\ B_f^T&A_i
\end{bmatrix}\begin{bmatrix}
x_f\\ x_l
\end{bmatrix}+\begin{bmatrix}
0_{n_f\times n_l}\\ I_{n_l}
\end{bmatrix}u,
\end{equation}
where $x_f=\begin{bmatrix}
x_1&x_2&\cdots&x_{n_f}
\end{bmatrix}^T,x_l=\begin{bmatrix}
x_{n_f+1}&\cdots&x_{n_f+n_l}
\end{bmatrix}^T  $ and $A_f=D_fD_f^T,B_f=D_fD_i^T,A_i=D_iD_i^T$. Multiplying with $D^T$ on both sides of \eqref{eq:DS_re}, we obtain the dynamics on the edge space as 
\begin{equation}\label{eq:DS_edge}
\Sigma_e: \dot {\bar{x}}=-L_e\bar{x}+D_i^Tu,
\end{equation}
with the edge Laplacian $L_e$. We know that $L_e$ is positive definite if the graph is a tree~\cite{dimarogonas2010stability}. We thus here assume  the following
\begin{assumption} 
The leader-follower multi-agent system \eqref{eq:nodedynamic} described by the graph $\mathcal{G}=(\mathcal{V},\mathcal{E})$ is a connected tree.
\end{assumption}
We consider tree graphs as a starting point since we need the positive definiteness of $L_e$ in the analysis, and motivated by the fact that they require less communication load (less edges) for their implementation. Note however that further results for a general graph could be built based on the results of tree graphs, for example, through graph decompositions~\cite{zelazo2011edge}. For the leader-follower multi-agent system \eqref{eq:DS_edge}, the proposed controller applied to the leader agents is the composition of the term based on prescribed performance of the positions of the neighbouring agents:
\begin{equation}\label{eq:control_node}
u_i=-\sum\limits_{j\in\mathcal{N}_i}g_{ij}\mathcal{J}_{T_{ij}}(\hat{x}_{ij},t)\varepsilon_{ij}(\hat{x}_{ij}), \hspace{5mm} i\in \mathcal{V}_L,
\end{equation} 
where $g_{ij}=g_{ji}$ is a positive scalar gain to be appropriately tuned. Then the stack input vector is 
\begin{equation}\label{eq:control}
u=-D_i\mathcal{J}_{T}(\hat{\bar{x}},t)G\varepsilon(\hat{\bar{x}}),
\end{equation} 
where $\hat{\bar{x}}$ is the stack vector of transformed errors $\hat{x}_{ij}$, $G\in \mathbb{R}^{m\times m}$ is a positive definite diagonal gain matrix with entries $g_{ij}$. $\mathcal{J}_T(\hat{\bar{x}},t)\in \mathbb{R}^{m\times m}$ is a time varying diagonal matrix with diagonal entries $\mathcal{J}_{T_{ij}}(\hat{x}_{ij},t)$, $\varepsilon(\hat{\bar{x}})\in \mathbb{R}^{m}$ is a stack vector with entries $\varepsilon_{ij}(\hat{x}_{ij})$. Then the edge dynamics \eqref{eq:DS_edge} with input \eqref{eq:control} can be written as
\begin{equation}\label{eq:DS_control}
\dot {\bar{x}}=-L_e\bar{x}-D_i^TD_i\mathcal{J}_{T}(\hat{\bar{x}},t)G\varepsilon(\hat{\bar{x}}),
\end{equation}

In the sequel, we develop the following result and will use Lyapunov-like methods to prove that the prescribed performance can be guaranteed and consensus can be achieved. 
\begin{theorem}
Consider the leader-follower multi-agent system $\Sigma$ under Assumption 1 with dynamics \eqref{eq:nodedynamic}, the predefined performance functions $\rho_{ij}$ as in \eqref{eq:performancefunction} and the transformation function s.t. $T_{ij}(0)=0,\forall (i,j)\in \mathcal{E}$, and assume that the initial conditions $x_{ij}(0)$ are within the performance bounds \eqref{eq:pfp} or \eqref{eq:pfn}. If the following condition holds:
\begin{equation} \label{eq:condition}
\bar{\gamma}\geq l= \max\limits_{(i,j)\in \mathcal{E}}( l_{ij}),
\end{equation}
where $l$ is the largest decay rate of $\rho_{ij}(t)$ and $\bar{\gamma}$ is the maximum value of $\gamma$ that ensures:
\begin{equation} \label{eq:bounds} 
\Gamma=\left[\begin{smallmatrix}
D_i^TD_i&\frac{1}{2}\left(L_e-\gamma\left(I_m-D_i^TD_i\right)\right)\\
\frac{1}{2}\left(L_e-\gamma\left(I_m-D_i^TD_i\right)\right)&\gamma L_e
\end{smallmatrix}\right]\geq 0.
\end{equation}
Then, the controlled system achieves consensus within the prescribed performance bounds $\rho_{ij}(t)$ when applying the control \eqref{eq:control}.
\end{theorem} 
\begin{proof}
Consider the Lyapunov-like function
\begin{equation}\label{eq:lyapu}
 V(\varepsilon_{\hat{\bar{x}}},\bar{x})=\frac{1}{2}\varepsilon_{\hat{\bar{x}}}^TG\varepsilon_{\hat{\bar{x}}}+\frac{\gamma}{2}\bar{x}^T\bar{x},   
\end{equation}
with $\varepsilon_{\hat{\bar{x}}}$ denoting $\varepsilon(\hat{\bar{x}})$ and $\mathcal{J}_{T_{\hat{\bar{x}}}}$ denoting $\mathcal{J}_{T}(\hat{\bar{x}},t)$. Then, $\dot{V}=\varepsilon_{\hat{\bar{x}}}^TG\dot{\varepsilon}_{\hat{\bar{x}}}+\gamma \bar{x}^T\dot{\bar{x}}$. Replacing $\dot{\varepsilon}_{\hat{\bar{x}}}$ according to \eqref{eq:diff}, we obtain
\begin{equation}
    \dot{V}=\varepsilon_{\hat{\bar{x}}}^TG\mathcal{J}_{T_{\hat{\bar{x}}}}(\dot{\bar{x}}+\alpha(t)\bar{x})+\gamma \bar{x}^T\dot{\bar{x}},
\end{equation}
where $\alpha(t)$ is the diagonal matrix with diagonal entries $\alpha_{ij}(t)$. According to \eqref{eq:alpha} and \eqref{eq:performancefunction}, we know that $\alpha_{ij}(t)< l_{ij}, \forall t$. Substituting \eqref{eq:DS_control}, we can further derive that


\begin{equation}\label{eq:lya1}
\begin{aligned}
\dot{V}=&\varepsilon_{\hat{\bar{x}}}^TG\mathcal{J}_{T_{\hat{\bar{x}}}}(-L_e\bar{x}-D_i^TD_i\mathcal{J}_{T_{\hat{\bar{x}}}}G\varepsilon_{\hat{\bar{x}}}+\alpha(t)\bar{x})\\&+\gamma \bar{x}^T(-L_e\bar{x}-D_i^TD_i\mathcal{J}_{T_{\hat{\bar{x}}}}G\varepsilon_{\hat{\bar{x}}})
\\=&-\varepsilon_{\hat{\bar{x}}}^TG\mathcal{J}_{T_{\hat{\bar{x}}}}L_e\bar{x}+\varepsilon_{\hat{\bar{x}}}^TG\mathcal{J}_{T_{\hat{\bar{x}}}}\alpha(t)\bar{x}\\&-\varepsilon_{\hat{\bar{x}}}^TG\mathcal{J}_{T_{\hat{\bar{x}}}}D_i^TD_i\mathcal{J}_{T_{\hat{\bar{x}}}}G\varepsilon_{\hat{\bar{x}}}-\gamma \bar{x}^TL_e\bar{x}\\&-\gamma \bar{x}^TD_i^TD_i\mathcal{J}_{T_{\hat{\bar{x}}}}G\varepsilon_{\hat{\bar{x}}}
\end{aligned} 
\end{equation}
Adding and subtracting  $\gamma\varepsilon_{\hat{\bar{x}}}^TG\mathcal{J}_{T_{\hat{\bar{x}}}}\bar{x}$ on the right hand side of \eqref{eq:lya1}, we obtain
\begin{equation}\label{eq:lya2}
\begin{aligned}
\dot{V}=&-\varepsilon_{\hat{\bar{x}}}^TG\mathcal{J}_{T_{\hat{\bar{x}}}}(\gamma I_m-\alpha(t))\bar{x}-\varepsilon_{\hat{\bar{x}}}^TG\mathcal{J}_{T_{\hat{\bar{x}}}}D_i^TD_i\mathcal{J}_{T_{\hat{\bar{x}}}}G\varepsilon_{\hat{\bar{x}}}\\&-\varepsilon_{\hat{\bar{x}}}^TG\mathcal{J}_{T_{\hat{\bar{x}}}}L_e\bar{x}-\gamma \bar{x}^TL_e\bar{x}+\gamma\varepsilon_{\hat{\bar{x}}}^TG\mathcal{J}_{T_{\hat{\bar{x}}}}(I_m-D_i^TD_i)\bar{x}\\=&-\varepsilon_{\hat{\bar{x}}}^TG\mathcal{J}_{T_{\hat{\bar{x}}}}(\gamma I_m-\alpha(t))\bar{x}
\\&-y^T\left[\begin{smallmatrix}
D_i^TD_i&\frac{1}{2}\left(L_e-\gamma\left(I_m-D_i^TD_i\right)\right)\\
\frac{1}{2}\left(L_e-\gamma\left(I_m-D_i^TD_i\right)\right)&\gamma L_e
\end{smallmatrix}\right]y
\\=&-\varepsilon_{\hat{\bar{x}}}^TG\mathcal{J}_{T_{\hat{\bar{x}}}}(\gamma I_m-\alpha(t))\bar{x}-y^T\Gamma y
\end{aligned}
\end{equation}
with
\begin{equation}
    y=\begin{bmatrix}
\mathcal{J}_{T_{\hat{\bar{x}}}}G\varepsilon_{\hat{\bar{x}}}\\\bar{x}
\end{bmatrix}.
\end{equation}
Since $G,\mathcal{J}_{T_{\hat{\bar{x}}}}$ are both diagonal and positive definite matrices, we have that $G\mathcal{J}_{T_{\hat{\bar{x}}}}$ is also a diagonal positive definite matrix. $(\gamma I_m-\alpha(t))$ is a diagonal positive definite matrix if $\gamma \geq l=\max(l_{ij})>\bar{\alpha}= \sup \alpha_{ij}(t)$. Due to $T_{ij}(0)=0$, we have $\varepsilon_{ij}(\hat{x}_{ij})\hat{x}_{ij}\geq 0$. Then, by by setting $\gamma:=\theta+\bar{\alpha}$, with $\theta$ being a positive constant we get:
\begin{equation}
    -\varepsilon_{\hat{\bar{x}}}^TG\mathcal{J}_{T_{\hat{\bar{x}}}}(\gamma I_m-\alpha(t))\bar{x}\leq -\theta \varepsilon_{\hat{\bar{x}}}^TG\mathcal{J}_{T_{\hat{\bar{x}}}}\bar{x}
\end{equation}
Then, according to \eqref{eq:modout}, \eqref{eq:jacobian}, we further obtain
\begin{equation} \label{eq:28}
    -\theta \varepsilon_{\hat{\bar{x}}}^TG\mathcal{J}_{T_{\hat{\bar{x}}}}\bar{x}=
    -\theta \varepsilon_{\hat{\bar{x}}}^TG\frac{\partial \varepsilon_{\hat{\bar{x}}}}{\partial \hat{\bar{x}}}\hat{\bar{x}}\leq 0.
\end{equation}
\eqref{eq:28} holds because the transformed function is smooth and strictly increasing and $\varepsilon_{ij}(\hat{x}_{ij})\hat{x}_{ij}\geq 0$. 
Therefore, in order for $\dot{V}\leq 0$ to hold, it suffices that $\gamma \geq l=\max(l_{ij})> \sup \alpha_{ij}(t)$ and in addition, $\Gamma$ should be semi-positive definite. Here, in order for $\Gamma \geq 0$ to be feasible, we need the assumption that the communication graph is a tree. This further means that $L_e$ is positive definite and \eqref{eq:bounds} is then equivalent to:
\begin{equation}
\begin{smallmatrix}
D_i^TD_i \geq  \frac{1}{4\gamma}\left(L_e-\gamma\left(I_m-D_i^TD_i\right)\right)L_e^{-1}\left(L_e-\gamma\left(I_m-D_i^TD_i\right)\right).
\end{smallmatrix}
\end{equation}
Then, based on condition \eqref{eq:condition}, and choosing $\gamma=\bar{\gamma}$, we obtain $-\varepsilon_{\hat{\bar{x}}}^TG\mathcal{J}_{T_{\hat{\bar{x}}}}(\bar{\gamma} I_m-\alpha(t))\bar{x}\leq 0$ and $\Gamma \geq 0$. Finally, we can conclude that $\dot{V}\leq 0$ when $\gamma=\bar{\gamma}$. This also implies $V(\varepsilon_{\hat{\bar{x}}},\bar{x})\leq V(\varepsilon_{\hat{\bar{x}}}(0),\bar{x}(0))$. Hence if $\bar{x}(0))$ is chosen within the region \eqref{eq:modpfp} or \eqref{eq:modpfn} then $V(\varepsilon_{\hat{\bar{x}}}(0),\bar{x}(0))$ is finite, which implies that $V(\varepsilon_{\hat{\bar{x}}},\bar{x})$ is bounded $\forall t$. Therefore $\varepsilon_{\hat{\bar{x}}},\bar{x}$ are bounded and the boundedness of the transformed error $\varepsilon_{\hat{\bar{x}}}$ implies that the relative position $\bar{x}(t)$ evolves within the prescribed performance bounds $\forall t$. Then we can prove the boundedness of $\ddot{V}(\varepsilon_{\hat{\bar{x}}},\bar{x})$ based on the boundedness of $\varepsilon_{\hat{\bar{x}}},\dot{\varepsilon}_{\hat{\bar{x}}}$. The boundedness of $\ddot{V}(\varepsilon_{\hat{\bar{x}}},\bar{x})$ implies the uniform continuity of $\dot{V}(\varepsilon_{\hat{\bar{x}}},\bar{x})$, which in turn implies that $\dot{V}(\varepsilon_{\hat{\bar{x}}},\bar{x})\rightarrow 0$ as $t\rightarrow \infty$ by applying Barbalat's Lemma. This implies $\bar{x}\rightarrow 0$ as $t\rightarrow \infty$ and consensus will be achieved.
\end{proof}


\begin{remark}
We are always interested in specifying the state of the multi-agent system at the equilibrium. Denote $x_c=\frac{1}{n}\sum\nolimits_{i=1}^{n}x_i$ as the centroid of the network. In most of the work regarding PPC like~\cite{macellari2017multi}, $\lim\limits_{t\rightarrow \infty} x_c(t)=x_c(0)=\frac{1}{n}\sum\nolimits_{i=1}^{n}x_i(0)$. This is because a PPC input for every agent exists. In our work, if we have an external input for every agent, i.e. $B=I_n$ in \eqref{eq:nodedynamic}, we  can also obtain $\lim\limits_{t\rightarrow \infty} x_c(t)=\frac{1}{n}\sum\nolimits_{i=1}^{n}x_i(0)$.
This can be verified by multiplying $\mathbf{1}^T$ on both sides of \eqref{eq:nodedynamic}, where $\mathbf{1}\in \mathbb{R}^n$ with all entries 1. Then, we can conclude $\dot{x}_c(t)=0$.
The main difference is that when we choose some leaders, we can achieve a varying equilibrium state of each agent by tuning the gain matrix, which is quite useful in practical design as we can decide where all the agents should gather.
\end{remark}
In the sequel, we will discuss the results for two specific classes of tree graphs: chain and star graph. First we consider the chain graph, which is wildly used for instance in autonomous vehicle platooning.  
\begin{definition}
A chain $\mathcal{G}^c=(\mathcal{V}^c,\mathcal{E}^c)$ is a tree graph with vertices set $\mathcal{V}^c=\{1,2,\dots,n\},n\geq 2$ and edges set $\mathcal{E}^c=\{(i,i+1)\in \mathcal{V}^c\times \mathcal{V}^c\mid i\in \mathcal{V}^c\setminus \{n\}\}$ indexed by $e_i=(i,i+1), i=1,2,\dots, n-1$.
\end{definition}
Note that \eqref{eq:condition} in Theorem 1 is a sufficient but not necessary condition. For a chain graph, the matrix inequality \eqref{eq:bounds} may be actually infeasible when the graph has 2 or more followers. The following result for $\mathcal{G}^c$ is derived.
\begin{proposition}
Consider the leader-follower multi-agent system $\Sigma$ described by \eqref{eq:nodedynamic} with the communication chain graph $\mathcal{G}^c=(\mathcal{V}^c,\mathcal{E}^c)$ and the followers set $\mathcal{V}_F^c=\{1,2,\dots,n_f\}$, the predefined performance functions $\rho_{ij}$ as in \eqref{eq:performancefunction} and the transformation function s.t. $T_{ij}(0)=0,\forall (i,j)\in \mathcal{E}$, and assume that the initial conditions $x_{ij}(0)$ are within the performance bounds \eqref{eq:pfp} or \eqref{eq:pfn}. Then, the chain can only have at most 3 followers ($n_f\leq 3$) in order to achieve consensus within the prescribed performance bounds $\rho_{ij}(t)$ when applying \eqref{eq:control}. Specifically, when the chain has 2 and 3 followers, 
\begin{equation} \label{eq:conditiontree}
\begin{aligned}
\max\limits_{(i,j)\in \mathcal{E}}(l_{ij})&=l\leq 2,\hspace{5mm} n_f=2;\\
\max\limits_{(i,j)\in \mathcal{E}}(l_{ij})&=l\leq 1,\hspace{5mm} n_f=3
\end{aligned}
\end{equation}
are the respective sufficient conditions under which the chain achieves consensus within the prescribed performance bounds $\rho_{ij}(t)$ when applying \eqref{eq:control}. 
\end{proposition}
\begin{proof}
When the chain graph has only one follower, that is $n_f=1$, the result can be proved by using Theorem 1. Let $\bar{\gamma}$ be the maximum value of $\gamma$ that ensures \eqref{eq:bounds} holds. By further choosing the decay rate of the performance functions \eqref{eq:performancefunction} to satisfy \eqref{eq:condition}, we can conclude that the controlled system achieves consensus within the prescribed performance bounds by applying \eqref{eq:control} based on Theorem 1. When the chain has additional followers, the condition in Theorem 1 may be infeasible since it is a sufficient but not necessary condition. But for this kind of special chain structure, we can resort to checking the edge dynamics \eqref{eq:DS_edge} directly. It can be shown that $-L_e$ has elements given by $c_{ij}=-2$ when $i=j$, $c_{ij}=1$ when $|i-j|=1$ and $c_{ij}=0$ otherwise when the graph is a chain. We then rewrite \eqref{eq:DS_edge} as 
\begin{equation}
\left[
	\begin{array}{c}
	\dot{\bar{x}}_f \\ \hline 
	\dot{\bar{x}}_l
	\end{array}
	\right]    =\left[
	\begin{array}{c|c}
	\mathbf{A}& \mathbf{B} \\ \hline 
	\mathbf{B}^T& \mathbf{C}
	\end{array}
	\right]\left[\begin{array}{c}
	\bar{x}_f \\ \hline 
	\bar{x}_l
	\end{array}
	\right] +\left[
	\begin{array}{c}
	\mathbf{0} \\ \hline 
	\mathbf{D}
	\end{array}
	\right]u, 
\end{equation}
where $\bar{x}_f\in \mathbb{R}^{(n_f-1)}$ represents the edges between followers, while $\bar{x}_l\in \mathbb{R}^{n_l}$ represents the edge that connects the leader node $\{n_f+1\}$ and the follower node $\{n_f\}$, and the edges between leaders. Both $\mathbf{A}\in \mathbb{R}^{(n_f-1)\times (n_f-1)},\mathbf{C}\in \mathbb{R}^{n_l\times n_l}$ have the same structure as $-L_e$ but with different dimensions, $\mathbf{B}$ has an element 1 at row $(n_f-1)$, column 1 (bottom left corner) that represents the connection between the  follower node $\{n_f\}$ and the leader node $\{n_f+1\}$. $\mathbf{0}$ is a $(n_f-1)\times n_l$ zero matrix. $\mathbf{D}\in \mathbb{R}^{n_l\times n_l}$ has elements given by $\mathbf{d}_{ij}=1$ when $i=j$, $\mathbf{d}_{ij}=-1$ when $i-j=1$ and $\mathbf{d}_{ij}=0$ otherwise. Then we can analyse the leader part $\bar{x}_l$ and the follower part $\bar{x}_f$ separately. For $\bar{x}_l$, it can be proved that $\bar{x}_l$ achieves consensus within the performance bounds based on the positive definiteness of $\mathbf{D}\mathbf{D}^T$ when applying control \eqref{eq:control}. We further rewrite the follower part as 
\begin{equation}\label{eq:followerpart}
    \dot{\bar{x}}_f=\mathbf{A}\bar{x}_f+\mathbf{b}\bar{x}_{\star},
\end{equation}
where $\mathbf{b}\in \mathbb{R}^{(n_f-1)}$ is the first column of $\mathbf{B}$, i.e., with the last element equals to 1 and all other elements equal to 0. $\bar{x}_{\star}$ represents the edge between the  follower node $\{n_f\}$ and the leader node $\{n_f+1\}$. We can furture solve the state evolution of \eqref{eq:followerpart} as follows:
\begin{equation}\label{eq:followerevolution}
\begin{aligned}
    \bar{x}_f(t)&=e^{\mathbf{A}t}\bar{x}_f(0)+\int_0^t e^{\mathbf{A}(t-\tau)}\mathbf{b}\bar{x}_{\star}(\tau)d\tau\\
    &=M^Te^{\mathbf{\Lambda}t}M\bar{x}_f(0)+\int_0^t e^{\mathbf{A}(t-\tau)}\mathbf{b}\bar{x}_{\star}(\tau)d\tau,\\
   &= \bar{x}^0_f(t)+\int_0^t e^{\mathbf{A}(t-\tau)}\mathbf{b}\bar{x}_{\star}(\tau)d\tau,
\end{aligned}
\end{equation}
where $\bar{x}^0_f(t)=\begin{bmatrix}\bar{x}^0_1(t)&\bar{x}^0_2(t)&\dots&\bar{x}^0_{n_f-1}(t) \end{bmatrix}^T$ is zero input trajectories, that is when $\bar{x}_\star(t)=0,\forall t$; $\mathbf{A}=M^T\mathbf{\Lambda}M$, where $\mathbf{\Lambda}$ is a diagonal matrix with diagonal entries negative and equal to the eigenvalues of $\mathbf{A}$, which is due to $\mathbf{A}$ having the same structure as $-L_e$, and $M$ is the matrix composed with the corresponding eigenvectors. Without loss of generality, suppose all performance functions are the same and described by
\begin{equation}\label{eq:performancefunction11} 
\rho(t)=(\rho_{0}-\rho_{\infty})e^{-lt}+\rho_{\infty}.
\end{equation}
When $n_f=2$, $\bar{x}_f=\bar{x}_1$ and $\mathbf{A}=-2$, we have that 
\begin{equation}
    \bar{x}^0_1(t)=M^Te^{\mathbf{\Lambda}t}M\bar{x}_1(0)=e^{-2t}\bar{x}_1(0)<\rho_0e^{-2t}. 
\end{equation}
Then, $\bar{x}_{1}(t)$ is within the performance bound $\rho(t)$, i.e., $\bar{x}_{1}(t)<\rho(t),\forall t$, when $l\leq 2$ and in addition,
\begin{equation}\label{eq:haha}
    \int_0^t e^{-2(t-\tau)}\bar{x}_{\star}(\tau)d\tau<(\rho_0-\bar{x}_1(0))e^{-2t}+ \rho_\infty(1-e^{-2t}),
\end{equation}
which can be ensured by tuning a large enough gain $g_{32}$ to the leader indexed by node $3$. From \eqref{eq:haha}, we know that when the relative position between the two followers is close to the boundary, we need to tune a larger gain for the leader that connects the followers. When $n_f=3$, we can derive a similar result. In particular, we now have that 
\begin{equation}
   \begin{bmatrix}\bar{x}^0_1(t)\\ \bar{x}^0_2(t)\end{bmatrix} =M^Te^{\mathbf{\Lambda}t}M\begin{bmatrix}\bar{x}_1(0)\\\bar{x}_2(0)\end{bmatrix}    <k\begin{bmatrix}\rho_0\\\rho_0\end{bmatrix}e^{-t},
\end{equation}
with $k=1$, which implies that $\bar{x}^0_i(t)<\rho_0 e^{-t}, i=\{1,2\}$. Similarly, we can conclude that when $l\leq 1$, and in addition the tuning gain $g_{43}$ for the leader indexed by node 4 is large enough, the controlled system achieves consensus within the prescribed performance bounds. When $n_f\geq 4$, it can be proved similarly that $\bar{x}^0_i(t)<k\rho_0 e^{\lambda_{\max}(\mathbf{A})t}, i=\{1,2,\dots,n_f-1\}$, but with $k>1$. This means that $\bar{x}^0_i(t)$ cannot be bounded by $\rho_0 e^{\lambda_{\max}(\mathbf{A})t}$ for any initial conditions within the performance bounds. Therefore, we can conclude that in order to achieve consensus within the performance bounds for all initial condition $x_{ij}(0)$ within the performance bounds \eqref{eq:pfp} or \eqref{eq:pfn}, $n_f$ should be less or equal to 3.
\end{proof}
\begin{remark}
Proposition 1 indicates that for a chain graph, in order to achieve consensus within the prescribed performance bounds, we can only have at most 3 consecutive followers at the end of the graph. In addition, when the initial relative position between 2 followers is close to the prescribed performance boundary, we need to tune a large enough gain for the leader that connects the followers.
\end{remark}
Now we consider another specific class, in particular the star graph $\mathcal{G}^s=(\mathcal{V}^s,\mathcal{E}^s)$ which is defined as follows.
\begin{definition}
A star $\mathcal{G}^s=(\mathcal{V}^s,\mathcal{E}^s)$ is a tree graph with vertices set $\mathcal{V}^s=\{1,2,\dots,n\},n\geq 2$ where vertice $n$ is called the centering node, and the edges set $\mathcal{E}^s=\{(i,n)\in \mathcal{V}^s\times \mathcal{V}^s\mid i\in \mathcal{V}^s\setminus \{n\}\}$ indexed by $e_i=(i,n), i=1,2,\dots, n-1$. 
\end{definition}
Then, the following result can be derived.
\begin{proposition}
Consider the leader-follower multi-agent system $\Sigma$ described by \eqref{eq:nodedynamic} with the communication star graph $\mathcal{G}^s=(\mathcal{V}^s,\mathcal{E}^s)$ and the leader set $\mathcal{V}_L^s=\{n\}$, the predefined performance functions $\rho_{ij}$ as in \eqref{eq:performancefunction} and the transformation function s.t. $T_{ij}(0)=0,\forall (i,j)\in \mathcal{E}$, and assume that the initial conditions $x_{ij}(0)$ are within the performance bounds \eqref{eq:pfp} or \eqref{eq:pfn}. If 
\begin{equation} \label{eq:conditionstar}
\max\limits_{(i,j)\in \mathcal{E}}(l_{ij})=l\leq 1.
\end{equation}
Then, the controlled system achieves consensus within the prescribed performance bounds $\rho_{ij}(t)$ when applying the control \eqref{eq:control}. 
\end{proposition}
\begin{proof}
For a star graph defined as Definition 2 with the centering node $n$ as the only leader, the edge Laplacian $L_e$ and matrices $D_i^TD_i, D_f^TD_f$ have special structures. $D_i^TD_i$ has all elements equal to 1, while $D_f^TD_f=L_e-D_i^TD_i$ is an identity matrix.  $L_e$ has the elements given by $c_{ij}=2$ when $i=j$, and $c_{ij}=1$ otherwise. Under this special structure of star graphs and according to Theorem 1, it can be verified that \eqref{eq:condition} is always feasible with $\bar{\gamma}=1$, and from \eqref{eq:conditionstar}, we know the condition $\bar{\gamma}\geq l= \max\limits_{(i,j)\in \mathcal{E}}( l_{ij})$ holds. Finally, by applying Theorem 1, for a star graph, when the performance functions \eqref{eq:performancefunction} are chosen such that \eqref{eq:conditionstar} holds, then  we can conclude that the controlled system achieves consensus within  the  prescribed  performance  bounds when applying \eqref{eq:control}.
\end{proof}
    We conclude this section with the following observations. A sufficient condition for a general tree graph was derived in Theorem 1, under which the leader-follower multi-agent system \eqref{eq:nodedynamic} achieves consensus within the prescribed performance bounds \eqref{eq:performancefunction}. It can be seen that \eqref{eq:condition} may be infeasible when the decay rate of the performance functions is too large. This means that we need to constrain the decay rate of the performance functions in order to achieve consensus under prescribed performance guarantees within the leader-follower framework. This is reasonable since the followers only obey the first-order consensus protocol without any additional external input. And the decay rate constraint differs for different graph topologies, leader amount and leader positions. For the specific class of star graphs, we have proven that when the largest decay rate of the performance functions is less than or equal to 1, the closed loop system achieves consensus within the prescribed performance bounds by applying Theorem 1. We have also shown that the condition in Theorem 1 is a sufficient but not necessary condition by discussing the specific class of chain graphs. That is, for a chain graph with 2 or 3 followers, we can still achieve the result of consensus within performance bounds although the condition in Theorem 1 may be infeasible.


\section{SIMULATIONS}
In this section three simulation examples are presented in order to verify the
results of the previous sections. The communication graphs are shown as Fig. \ref{tree}, where the leaders and followers are represented by grey and white nodes, respectively. Regarding the prescribed performance functions, for all $(i,j)\in \mathcal{E}$, we choose $M_{ij}=1$ and
\begin{equation}\nonumber
     T_{ij}(\hat{x}_{ij})=\ln \left(-\frac{\hat{x}_{ij}+1}{\hat{x}_{ij}-1}\right). 
\end{equation}
The prescribed performance bounds are chosen as in \eqref{eq:funcpara} with different decay rate $l$ for different simulation examples. For each graph, choosing the same $\rho_{ij}$ for all edges is done without loss of generality. In addition, the prescribed performance bounds are depicted in black color for the following simulation graphs.
\begin{equation} \label{eq:funcpara}
\rho_{ij}(t)=4.9e^{-lt}+0.1.
\end{equation}
\begin{figure}[!h]
\centering
\includegraphics[width=1\columnwidth]{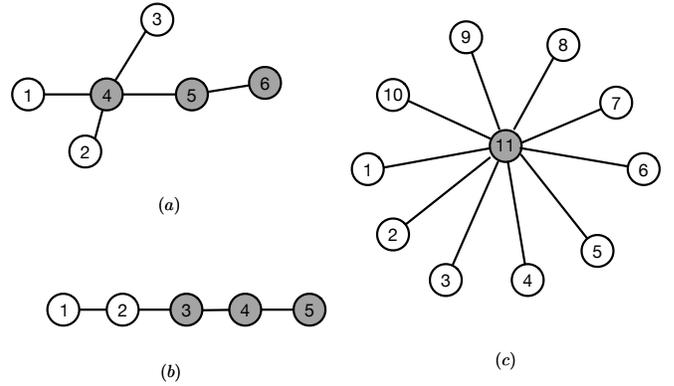}
\caption{Communication graphs with tree topologies.}
\label{tree}
\end{figure}

In Fig. \ref{tree}.(a), We first consider a tree graph with leaders set as $\mathcal{V}_L=\{4,5,6\}$, and the relative positions are initialised as $\begin{bmatrix}4.6& 4.9 &4.5& 4.7& 4.5\end{bmatrix}^T$. According to Theorem 1, the matrix inequality is feasible with $\bar{\gamma}=1$, hence it suffices that $l\leq \bar{\gamma}= 1$. The simulation result when applying the PPC law \eqref{eq:control} with a gain matrix $G$ whose diagonal entries are all equal to $1$ is shown on the right side of Fig. \ref{sim1}. As a comparison, the simulation result without PPC is shown on the left side of Fig. \ref{sim1}. We can see from Fig. \ref{sim1} that the trajectories intersect the performance bound without extra control, which can be improved by applying the PPC law \eqref{eq:control} such that the controlled system achieves consensus within the performance bound. Here the decay rate of the prescribed performance function is 1.

\begin{figure}[!h]
\centering
\includegraphics[width=1\columnwidth]{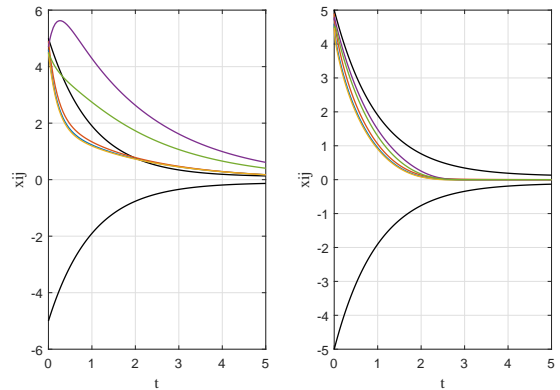}
\caption{The left figure shows the trajectories of relative positions without PPC, while the controlled system with PPC law is shown in the right figure under the communication graph as in Fig. \ref{tree}.(a).}
\label{sim1}
\end{figure}

In Fig. \ref{tree}.(b), we consider a chain graph with followers set as $\mathcal{V}_F=\{1,2\}$ and $\mathcal{V}_F=\{1,2,3\}$, the relative positions are initialised as $\begin{bmatrix}4.8& 3 &-2& 1\end{bmatrix}^T$. When the system has 2 followers, we know that the performance function can have a higher decay rate of 2, while the maximum decay rate is 1 when the system has one more follower (agent 3).
When $\mathcal{V}_F=\{1,2\}$, the simulation results are shown in Fig. \ref{sim2}, where the left figure shows the simulation result without additional control. Here the decay rate of the prescribed performance function is 2. We can see that the trajectories intersect the performance bound, which is improved as shown in the middle figure by applying the PPC law \eqref{eq:control} with gain matrix $G=diag(1,10,1,1)$, where $diag(a_1,a_2,\dots,a_n)$ represents the diagonal matrix with diagonal entries $a_1,a_2,\dots,a_n$ and $g_{32}=10$ is tuned for leader $\{3\}$ that connects the followers. However, it can be seen that the trajectories still intersect the performance bound. We then increase $g_{32}$ to $200$, and the simulation result is shown in the right figure. We can see that the controlled system achieves consensus within the performance bound. When $\mathcal{V}_F=\{1,2,3\}$, the simulation results are shown as in Fig. \ref{sim3}, in which the decay rate of the prescribed performance function is 1. Similarly, it can be seen in the left figure that the trajectories intersect the performance bound when there is no extra input, which is improved as shown in the middle and right figure by applying the PPC law \eqref{eq:control} with gain matrix $G=diag(1,1,10,1)$ and $G=diag(1,1,100,1)$, respectively. Here, the large gain is tuned for agent $4$ because it is the leader that connects the followers. We can also conclude that the controlled system achieves consensus within the performance bound.       

\begin{figure}[!h]
\centering
\includegraphics[width=1\columnwidth]{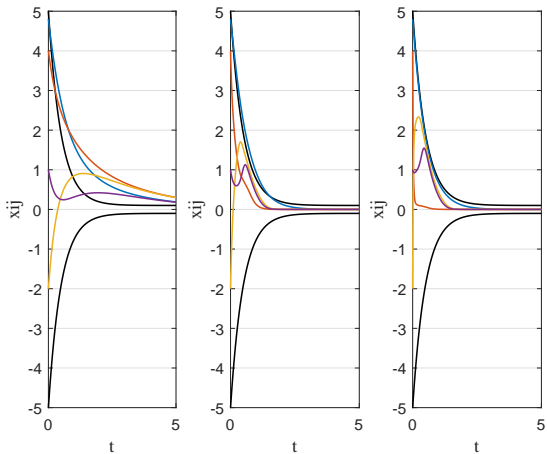}
\caption{The left figure shows the trajectories of relative positions without PPC, while the controlled system with PPC law but different gain matrix is shown in the middle and right figure, respectively under the communication graph as in Fig. \ref{tree}.(b) with $\mathcal{V}_F=\{1,2\}$.}
\label{sim2}
\end{figure}
\begin{figure}[!h]
\centering
\includegraphics[width=1\columnwidth]{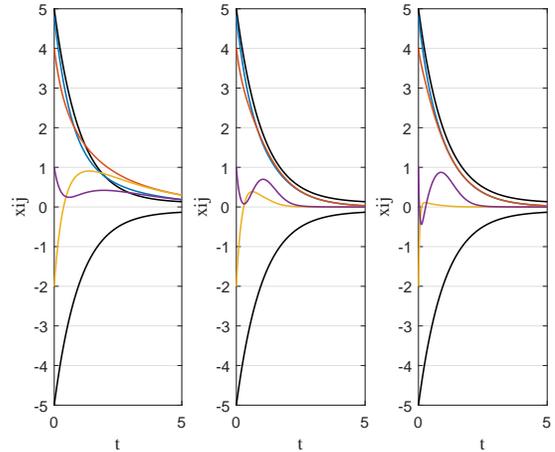}
\caption{The left figure shows the trajectories of relative positions without PPC, while the controlled system with PPC law but different gain matrix is shown in the middle and right figure, respectively under the communication graph as in Fig. \ref{tree}.(b) with $\mathcal{V}_F=\{1,2,3\}$.}
\label{sim3}
\end{figure}
In Fig. \ref{tree}.(c), We consider a star graph with only one leader as $\mathcal{V}_L=\{11\}$, and the relative positions are initialised as $\begin{bmatrix}4& 3 &-2& -3&4.9&1&4.7&-4&1&4.8\end{bmatrix}^T$. The simulation result when applying PPC law \eqref{eq:control} with a gain matrix $G$ whose diagonal entries are all equal to 1 is shown on the right side of Fig. \ref{sim4}. As a comparison, the simulation result without PPC is shown on the left side of Fig. \ref{sim4}. It is shown that the trajectories intersect the performance bound when there is no extra input, which can be improved by applying the PPC law \eqref{eq:control} such that the controlled system achieves consensus within the performance bound. Here the decay rate of the prescribed performance function is 1.

\begin{figure}[!h]
\centering
\includegraphics[width=1\columnwidth]{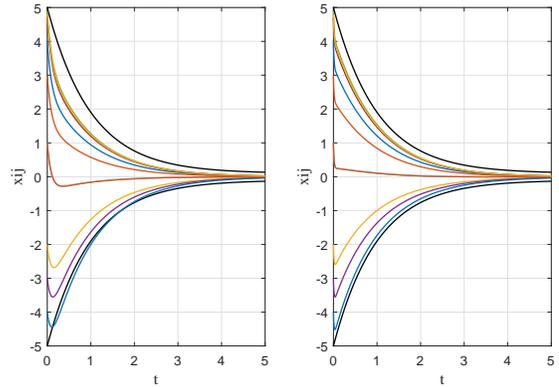}
\caption{The left figure shows the trajectories of relative positions without PPC, while the controlled system with PPC law is shown in the right figure under the communication graph as in Fig. \ref{tree}.(c).}
\label{sim4}
\end{figure}


\section{CONCLUSIONS}
In this paper, we have studied consensus problems of leader-follower multi-agent systems with prescribed performance bounds. Under the assumption of tree graphs, a distributed prescribed performance control law has been proposed for a group of selected leaders in order to drive the followers such that the entire system can achieve consensus under the prescribed performance guarantees. We have proved that when the decay rate of the performance functions is within a sufficient bound, consensus together with performance guarantees can be obtained. In addition, two specific classes of chain and star graphs that  can  have  additional followers have been investigated.

Future research directions include considering more general graphs with circles, applying other transient approaches to this leader-follower framework and also investigating leader selection problems.



\begin{thebibliography}{10}

\bibitem{balch1998behavior}
T.~Balch and R.~C. Arkin.
\newblock Behavior-based formation control for multirobot teams.
\newblock {\em IEEE transactions on robotics and automation}, 14(6):926--939,
  1998.

\bibitem{bechlioulis2014robust}
C.~P. Bechlioulis and K.~J. Kyriakopoulos.
\newblock Robust model-free formation control with prescribed performance and
  connectivity maintenance for nonlinear multi-agent systems.
\newblock In {\em 53rd IEEE Conference on Decision and Control}, pages
  4509--4514. IEEE, 2014.

\bibitem{bechlioulis2008robust}
C.~P. Bechlioulis and G.~A. Rovithakis.
\newblock Robust adaptive control of feedback linearizable mimo nonlinear
  systems with prescribed performance.
\newblock {\em IEEE Transactions on Automatic Control}, 53(9):2090--2099, 2008.

\bibitem{berger2018funnel}
T.~Berger, H.~H. L{\^e}, and T.~Reis.
\newblock Funnel control for nonlinear systems with known strict relative
  degree.
\newblock {\em Automatica}, 87:345--357, 2018.

\bibitem{dimarogonas2010stability}
D.~V. Dimarogonas and K.~H. Johansson.
\newblock Stability analysis for multi-agent systems using the incidence
  matrix: quantized communication and formation control.
\newblock {\em Automatica}, 46(4):695--700, 2010.

\bibitem{egerstedt2012interacting}
M.~Egerstedt, S.~Martini, M.~Cao, K.~Camlibel, and A.~Bicchi.
\newblock Interacting with networks: How does structure relate to
  controllability in single-leader, consensus networks?
\newblock {\em IEEE Control Systems}, 32(4):66--73, 2012.

\bibitem{fax2003information}
J.~A. Fax and R.~M. Murray.
\newblock Information flow and cooperative control of vehicle formations.
\newblock 2003.

\bibitem{franchi2018online}
A.~Franchi and P.~R. Giordano.
\newblock Online leader selection for improved collective tracking and
  formation maintenance.
\newblock {\em IEEE transactions on control of network systems}, 5(1):3--13,
  2018.

\bibitem{goldin2010controllability}
D.~Goldin and J.~Raisch.
\newblock Controllability of second order leader-follower systems.
\newblock In {\em 2nd IFAC Workshop on Distributed Estimation and Control in
  Networked Systems 2010-NecSys 10}, pages 233--238, 2010.

\bibitem{ilchmann2005tracking}
A.~Ilchmann, E.~P. Ryan, and S.~Trenn.
\newblock Tracking control: Performance funnels and prescribed transient
  behaviour.
\newblock {\em Systems \& Control Letters}, 54(7):655--670, 2005.

\bibitem{karayiannidis2012multi}
Y.~Karayiannidis, D.~V. Dimarogonas, and D.~Kragic.
\newblock Multi-agent average consensus control with prescribed performance
  guarantees.
\newblock In {\em Decision and Control (CDC), 2012 IEEE 51st Annual Conference
  on}, pages 2219--2225. IEEE, 2012.

\bibitem{katsoukis2016output}
I.~Katsoukis and G.~A. Rovithakis.
\newblock Output feedback leader-follower with prescribed performance
  guarantees for a class of unknown nonlinear multi-agent systems.
\newblock In {\em 2016 24th Mediterranean Conference on Control and Automation
  (MED)}, pages 1077--1082. IEEE, 2016.

\bibitem{macellari2017multi}
L.~Macellari, Y.~Karayiannidis, and D.~V. Dimarogonas.
\newblock Multi-agent second order average consensus with prescribed transient
  behavior.
\newblock {\em IEEE Transactions on Automatic Control}, 62(10):5282--5288,
  2017.

\bibitem{mesbahi2010graph}
M.~Mesbahi and M.~Egerstedt.
\newblock {\em Graph theoretic methods in multiagent networks}, volume~33.
\newblock Princeton University Press, 2010.

\bibitem{olfati2004consensus}
R.~Olfati-Saber and R.~M. Murray.
\newblock Consensus problems in networks of agents with switching topology and
  time-delays.
\newblock {\em IEEE Transactions on automatic control}, 49(9):1520--1533, 2004.

\bibitem{patterson2010leader}
S.~Patterson and B.~Bamieh.
\newblock Leader selection for optimal network coherence.
\newblock In {\em 49th IEEE Conference on Decision and Control (CDC)}, pages
  2692--2697. IEEE, 2010.

\bibitem{rahmani2009controllability}
A.~Rahmani, M.~Ji, M.~Mesbahi, and M.~Egerstedt.
\newblock Controllability of multi-agent systems from a graph-theoretic
  perspective.
\newblock {\em SIAM Journal on Control and Optimization}, 48(1):162--186, 2009.

\bibitem{ren2007distributed}
W.~Ren and E.~Atkins.
\newblock Distributed multi-vehicle coordinated control via local information
  exchange.
\newblock {\em International Journal of Robust and Nonlinear Control:
  IFAC-Affiliated Journal}, 17(10-11):1002--1033, 2007.

\bibitem{tanner2004controllability}
H.~G. Tanner.
\newblock On the controllability of nearest neighbor interconnections.
\newblock In {\em Decision and Control, 2004. CDC. 43rd IEEE Conference on},
  volume~3, pages 2467--2472. IEEE, 2004.

\bibitem{tanner2007flocking}
H.~G. Tanner, A.~Jadbabaie, and G.~J. Pappas.
\newblock Flocking in fixed and switching networks.
\newblock {\em IEEE Transactions on Automatic control}, 52(5):863--868, 2007.

\bibitem{yaziciouglu2013leader}
A.~Y. Yazicio{\u{g}}lu and M.~Egerstedt.
\newblock Leader selection and network assembly for controllability of
  leader-follower networks.
\newblock In {\em American Control Conference (ACC), 2013}, pages 3802--3807.
  IEEE, 2013.

\bibitem{zelazo2011edge}
D.~Zelazo and M.~Mesbahi.
\newblock Edge agreement: Graph-theoretic performance bounds and passivity
  analysis.
\newblock {\em IEEE Transactions on Automatic Control}, 56(3):544--555, 2011.

\end{thebibliography}






\end{document}